\documentclass[11pt,letterpaper,english]{article}

\usepackage{amssymb}
\usepackage{amsmath}
\usepackage{amsthm}
\usepackage[ruled,vlined]{algorithm2e}
\usepackage[dvipsnames]{xcolor}
\usepackage{booktabs}
\usepackage{tikz}
\usepackage{pgfplots}
\pgfplotsset{compat=newest}

\usetikzlibrary{shapes}
\usepackage{todonotes}
\usepackage{mathtools}
\usepackage{graphicx}
\usepackage{authblk}

\newtheorem{theorem}{Theorem}
\newtheorem{lemma}[theorem]{Lemma}
\newtheorem{proposition}[theorem]{Proposition}

\newcommand{\RR}[0]{\mathbb{R}}
\renewcommand{\SS}[0]{\mathcal{S}}
\newcommand{\OPT}{\mathrm{OPT}}
\newcommand{\ALG}{\mathrm{ALG}}

\newcommand{\e}{\mathrm e}
\newcommand{\bigO}{\mathcal{O}}

\newcommand{\requests}{\mathcal{R}}
\newcommand{\C}{\mathcal{C}}

\renewcommand{\vec}[1]{\mathbf{#1}}

\begin{document}
\title{The Online Best Reply Algorithm for Resource Allocation Problems}
%
%
\author[1]{Max Klimm} 
\author[2]{Daniel Schmand} 
\author[3]{Andreas T{\"o}nnis\thanks{Partially supported by CONICYT grant PCI PII 20150140 and ERC Starting Grant 306465 (BeyondWorstCase).}}

\affil[1]{Operations Research, HU Berlin, Germany.\authorcr max.klimm@hu-berlin.de}
\affil[2]{Institute for Computer Science, Goethe University Frankfurt, Germany.\authorcr schmand@em.uni-frankfurt.de}
\affil[3]{Department of Computer Science, University of Bonn, Germany.\authorcr atoennis@uni-bonn.de}

\maketitle              
\begin{abstract}
We study the performance of a best reply algorithm for online resource allocation problems with a diseconomy of scale. In an online resource allocation problem, we are given a set of resources and a set of requests that arrive in an online manner. Each request consists of a set of feasible allocations and an allocation is a set of resources. The total cost of an allocation vector is given by the sum of the resources' costs, where each resource's cost depends on the total load on the resource under the allocation vector. We analyze the natural online procedure where each request is allocated greedily to a feasible set of resources that minimizes the individual cost of that particular request. In the literature, this algorithm is also known as a one-round walk in congestion games starting from the empty state. For unweighted resource allocation problems with polynomial cost functions with maximum degree $d$, upper bounds on the competitive ratio of this greedy algorithm were known only for the special cases $d\in\{1, 2, 3\}$. In this paper, we show a general upper bound on the competitive ratio of $d(d / W(\frac{1.2d-1}{d+1}))^{d+1}$ for the unweighted case where $W$ denotes the Lambert-W function on $\mathbb{R}_{\geq 0}$. For the weighted case, we show that the competitive ratio of the greedy algorithm is bounded from above by $(d/W(\frac{d}{d+1}))^{d+1}$.

\end{abstract}


\section{Introduction}

We consider a greedy best reply algorithm for online resource allocation problems. The set of feasible allocations for each request is a set of subsets of the resources. Each resource is endowed with a cost function that is a polynomial with non-negative coefficients depending on the total load of that resource.
In the online variant considered in this paper, the requests arrive one after another. Upon arrival of a request, we immediately and irrevocably choose a feasible allocation for that request without any knowledge about requests arriving in the future. After the sequence of requests terminates, we evaluate the solution quality of the best reply algorithm in terms of its \emph{competitive ratio}  defined as the worst-case over all instances of the ratio of the cost of an online solution and the cost of an offline optimal solution. Here, the cost of a solution is defined as the sum of the resources' cost under the allocation vector. The cost of a resource is the sum of the personal costs of each request on that resource. Specifically, we consider the natural greedy best reply algorithm that assigns each request to the allocation that minimizes the personal cost of the request. More formally, in an unweighted resource allocation problem all requests have a unit weight, the cost of each resource depends on the number of requests using it. In a weighted resource allocation problem, each request~$i$ has a weight $w_i$, the personal cost of each request on the resource depends on the total weight of requests using it.

A prominent application of this model is energy efficient algorithm design. Here, resources model machines or computing devices that can run at different speeds. A sequence of jobs is revealed in an online manner and has to be allocated to a set of machines such that all machines process the tasks allocated to them within a certain time limit. As a consequence, machines have to run at higher speed when more tasks are are allocated to the machine. As its speed is increased, the energy consumption of a machine increases superlinearly; a typical assumption in the literature is that the energy consumption is a polynomial with non-negative coefficients and maximal degree~3 as a function of the load \cite{Albers}. The aim is to find an allocation with minimal energy consumption. Our results imply bounds on the competitive ratio of the natural online algorithm that assigns each task to a (set of machines) that greedily minimizes the energy consumption of that task.

Another application of the resource allocation problems considered in this paper arises in the context of congestion games. Here, requests correspond to players and feasible allocations correspond to feasible strategies of that player. In a network congestion game, e.g., the set of strategies of a player is equal to the set of paths from some designated start node to a designated destination node in a given graph. Congestion occurs on links that are chosen by multiple users and is measured in terms of a load-dependent cost function.
Polynomial cost functions play a particularly important role in the field of traffic modeling; a typical class of cost functions used in traffic models are polynomials with non-negative coefficients and maximal degree~4 as proposed by the US Bureau of Public Roads~\cite{usPublicRoads}. 
The online variant of the resource allocation problem models the situation where users arrive one after another and choose a path that minimizes their private cost with respect to the users that are already present. This scenario is very natural, e.g., in situations where requests to a supplier of connected automotive navigation systems appear online and requests are routed such that the travel time for each request is minimized individually.
Our results imply bounds on the competitive ratio of the natural greedy algorithm were each player chooses a strategy that minimizes their total cost given the set of players already present in the network.

\subsection{Related Work}
\label{sec:related_work}
Already the approximation of optimal solutions to the offline version of resource allocation problems with polynomial cost functions is very challenging. Roughgarden~\cite{Roughgarden14} showed that there is a constant $\beta > 0$ such that the optimal solution cannot be approximated in polynomial time by a factor better than $(\beta d)^{d/2}$ when cost functions are polynomials of maximum degree $d$ with non-negative coefficients. This holds even for the unweighted case. For arbitrary cost functions, the optimal solution cannot be approximated by a constant factor in polynomial time \cite{Meyers2012}.
For polynomials of maximal degree $d$, currently the best known approximation algorithm is due to Makarychev and Srividenko~\cite{MakarychevS14} and uses a convex programming relaxation. They showed that randomly rounding an optimal fractional solution gives an $\mathcal{O}\big((\frac{0.792d}{\ln (d+1)})^{d}\big)$ approximate solution. This approach is highly centralized and relies on the fact that all requests are initially known, which both might be unrealistic assumptions for large-scale problems. 

Online and decentralized algorithms that have been studied in the literature are local search algorithms and multi-round best-reply dynamics. The analysis of both algorithms is technically very similar to the now-called \emph{smoothness} technique to establish bounds on the price of anarchy of Nash equilibria in congestion games \cite{AlandDGMS2011,AwerbuchAE13,Roughgarden15smoothness,ChristodoulouK05}. The price of anarchy is equal to the worst-case ratio of the cost of a Nash equilibrium and the cost of an optimal solution. 
To obtain tight bounds, one solves an optimization problem of the form
\[\min_{\lambda >0, \mu \in [0,1)} \Big\{\frac{\lambda}{1-\mu} : c(x+y)\leq \lambda x c(x) + \mu y c(y), \forall x,y \in \mathbb{N}, c \in \mathcal{C}\Big\}\;,\] where $\mathcal{C}$ is the set of cost functions in the game. For the case that $\mathcal{C}$ is the set of polynomial functions with maximum degree $d$ and positive coefficients, Aland et al.~\cite{AlandDGMS2011} used this approach to show that the price of anarchy is $\Phi_d^{d+1}$ where $\smash{\Phi_d \in\Theta(\frac{d}{\ln(d)})}$ is the unique solution to $\smash{(x+1)^d=x^{d+1}}$.
The price of stability, defined as the worst case of the ratio of the cost of a best Nash equilibrium and that of a system optimum, was not as well understood  until recently when Christodoulou et al.~\cite{ChristodoulouGGS18} showed that the price of stability is at least $(\Phi_d/2)^{d+1}$ for large $d$. For unweighted congestion games, Christodoulou and Gairing \cite{ChristodoulouG16} showed a tight bound on the price of stability in the order of $\Theta(d)$. Unfortunately, a best-response walk towards a Nash equilibrium can take exponential time~\cite{FabrikantPT04}, even for unweighted congestion games, so that price of anarchy results do not give rise to polynomial approximation algorithms. For weighted games, best-response walks may even cycle \cite{Harks2012}. On the other hand, random walks~\cite{GoemansMV05} or walks using approximate best-response steps~\cite{AwerbuchAEMS08} converge to approximate Nash equilibria in polynomial time.
In contrast to this, one-round walks in congestion games, or equivalently, the best reply algorithm for online resource allocation problems touches every request only once. Fanelli et al.~\cite{FanelliFM12} have shown a linear lower bound even for linear cost functions if the requests are restricted to make one best-response starting from a bad initial configuration. This lower bound does not hold for one round walks starting in the empty state.

The best reply algorithm with respect to the personal cost of a request, which we analyze here, has also been studied for the online setting before. For weighted resource allocation problems with linear cost functions, it turns out that the algorithm admits the same competitive ratio as the best reply algorithm with respect to the total cost function. There is a tight bound on the competitive ratio of $3 + 2\sqrt{2} \approx 5.83$, where the lower bound is due to Caragiannis et al.~\cite{caragiannis2011tight} and the upper bound is due to Harks et al.~\cite{HarksHPV07}. For $d > 1$, a first lower bound of $\Omega((d / \ln 2)^{d+1})$ has been shown by Caragiannis et al.~\cite{Caragiannis08}. A first upper bound dates back to the mid 90s when Awerbuch et al.~\cite{AwerbuchAGKKV95} gave an upper bound on the competitive ratio of personal cost best replies of $\Psi_d^{d+1}$, where $\Psi_d$ is defined to be the unique solution to the equation $(d+1)(x+1)^d = x^{d+1}$. However, they only considered the setting of singleton requests where each allocation contains a single resource only and all cost functions are equal to the identity. Bil{\`o} and Vinci~\cite{BiloV17} show that the worst-case competitive ratio is in fact attained for singletons and mention that the tight competitive ratio is $\Psi_d^{d+1}$, but their paper does not contain a proof of the latter result. 
Prior to that, Harks et al.~\cite{HarksHPV07} noted that the competitive ratio in the order of $\bigO(1.77^d d^{d+1})$. We here slightly improve the bound to a closed form as we note that $\Psi_d \leq d / W(\frac{d}{d+1})$, where $W$ is the Lambert-W function. This recovers the bound by Harks et al.\ in the limit since $1/W(\frac{d}{d+1}) \approx 1.77$ for $d$ large enough.

For unweighted instances it turns out that the personal cost best reply algorithm admits a better competitive ratio than the total cost best reply algorithm, where requests are allocated greedily such that the total cost of the current solution is minimized. For $d=1$, the tight bound is $\smash{\frac{(\phi +1)^2}{\phi} \approx 4.24}$ where $\smash{\phi = \frac{1+\sqrt{5}}{2}}$ is the golden ratio. The lower bound is due to Bil\`o et al.~\cite{bilo2011performance} and the upper bound is due to Christodoulou et al.~\cite{ChristodoulouMS12}. For arbitrary $d$, the lower bound of $(d+1)^{d+1}$ by Farzad et al.~\cite{Farzad08} also holds in this setting. There was no general upper bound known.

Harks et al.~\cite{HarksHP09,HarksHPV07} and Farzad et al.~\cite{Farzad08} both also studied a setting that is equivalent to ours. However, they measure cost slightly different as they define the cost of a resource as the integral of its cost function from $0$ to the current load. This different cost measure leads also to a different notion of the private cost of a request. However, as remarked by Farzad et al.~\cite{Farzad08}, the models are equivalent when only polynomial costs are considered. Additionally, they introduced a generalization of the model, where requests are only present during certain time windows. It is easy to see that all our results also hold in this slightly more general setting.

Close to our work is the analysis of best reply algorithms with different personal cost functions. Mirrokni and Vetta~\cite{MirrokniV04} were the first to study best-response dynamics with respect to social cost. Bjelde et al.~\cite{Bjelde2017approximation} analyzed the solution quality of local minima of the total cost function both for weighted and unweighted resource allocation problems. By a result of Orlin et al.\ \cite{orlinPunnenSchulz}, this admits a PTAS in the sense that an (1 + $ \epsilon$)-approximate local optimal
solution can be computed in polynomial time via local improvement steps. 
Best reply algorithms with respect to the social cost function instead of the personal cost function have also been well studied for the online setting. For weighted resource allocation problems and $d=1$, there is a tight bound of 5.83. The upper bound is due to Awerbuch et al.~\cite{AwerbuchAGKKV95} and the lower bound due to Caragiannis et al.~\cite{caragiannis2011tight}. For larger $d$, there is a lower bound of $\Omega((d / \ln 2)^{d+1})$ by Caragiannis et al.~\cite{Caragiannis08} and an upper bound of $\mathcal{O}((d / \ln 2)^{d+1})$ by Bil\`o and Vinci~\cite{BiloV17}. For unweighted resource allocation problems there is a tight bound of 5.66 in the linear case. The upper bound is due to Suri et al.~\cite{suri2007selfish} and the lower bound due to Caragiannis et al.~\cite{caragiannis2011tight}. For larger $d$, there is a lower bound of $(d+1)^{d+1}$ by Farzad et al.~\cite{Farzad08}. Up to our knowledge, there is no known upper bound that separates the unweighted case from the weighted case.
\subsection{Our Contribution}

We show upper bounds on the competitive ratio of the best reply algorithm in online resource allocation problems with cost functions that are polynomials of maximal degree $d$ with non-negative coefficients. 
%
For unweighted instances, we provide the first bound that hold for any fixed value of $d$. Prior to our results, non-trivial upper bounds were only known for the cases $d=1$ by Christodoulou et al.~\cite{ChristodoulouMS12} and for the cases $d \in \{2,3\}$ by Bil\`o~\cite{DBLP:journals/mst/Bilo18}. To the best of our knowledge, despite the wealth of results for weighted problems, prior to this work, no competitive ratio for any $d>3$ or the asymptotic behavior as $d \to \infty$ has been known that holds specifically for the unweighted case. We close this gap and show that the best reply algorithm is $d(\Xi_dd)^{d+1}$ competitive, where $\Xi_d \leq 1 / W(\frac{1.2d-1}{d+1})$. Here $W$ is the Lambert-W function on $\mathbb{R}_{\geq 0}$. Thus, we obtain a concrete factor that holds for any $d$. We further show that $\lim_{d \rightarrow \infty}{\Xi_d} \approx 1.523$ thus also giving the asymptotic behavior of the bound.

For weighted resource allocation problems, previous work \cite{AwerbuchAGKKV95,BiloV17} has established an upper bound of $\Psi_d^{d+1}$, where $\Psi_d$ is defined to be the solution to the equation $(d+1)(x+1)^d = x^{d+1}$. We choose to add a short proof of this result in the appendix since the proof of Awerbuch et al.~\cite{AwerbuchAGKKV95} does only consider singletons with identical resource cost functions and the paper by Bil{\`o} and Vinci~\cite{BiloV17} mentions the result without a proof. We also show that $\Psi_d \leq d / W(\frac{d}{d+1})$. This refines an upper bound of $\mathcal{O}(1.77^dd^{d+1})$ obtained by Harks et al.~\cite{HarksHPV07}. Note that in the limit, both results coincide as $\lim_{d\rightarrow\infty}\Psi_d \approx 1.77$.

Both our main result concerning unweighted games and our proofs for weight\-ed games allow for the first time to separate the competitive ratios of greedy personal best replies for unweighted and weighted problems, respectively. While for weighted games the competitive ratio is about $(1.77d)^{d+1}$, for unweighted games it is bounded from above by $(1.523d)^{d+1}$ for $d$ large enough.

Due to space constraints, some proofs are deferred to the appendix.

\section{Preliminaries}

We consider online algorithms for unsplittable resource allocation problems. Let $R$ be a finite set of resources $r$ each endowed with a non-negative cost function $c_r : \RR_{\geq 0} \to \RR_{\geq 0}$. There is a sequence of requests $\requests = (w_1, \SS_1), \dots, (w_n, \SS_n)$. At time step $i$, the existence of request $(w_i, \SS_i)$ is revealed, where $w_i$ is its weight, and $\SS_i \subseteq 2^R$ is the set of feasible allocations. If $w_i=1$ for all $i \in \{1,\dots,n\}$, we call the instance \emph{unweighted}. Upon arrival of request~$i$, an allocation $S_i \in \SS_i$ has to be fixed irrevocably by an online algorithm.

We use the notation $[n]=\{1,\dots,n\}$. For $i \in [n]$, let $\vec{\SS}_{\leq i} = \SS_1 \times \dots \times \SS_{i-1} \times \SS_i$ be the set of all feasible allocation vectors up to request~$i$. For a resource $r \in R$ and an allocation vector $\vec S_{\leq i} = (S_1,\dots,S_{i-1},S_i) \in \vec{\SS}_{\leq i}$ we denote the load of $r$ under $\vec S_{\leq i}$ by $w_r(\vec S_{\leq i})$. In the following, we write $\SS$ and $\vec S$ instead of $\SS_{\leq n}$ and $\vec S_{\leq n}$. The total cost of an allocation vector $\vec S$ is defined as \[C(\vec S) = \sum_{i \in [n]}\sum_{r \in S_i}w_ic_r(w_r(\vec S)) = \sum_{r \in R}w_r(\vec S)c_r(w_r(\vec S))\;.\]
Given a sequence of requests $\requests$, the offline optimal solution value is denoted by $\OPT(\requests) = \min_{\vec S \in \vec{\SS}} C(\vec S)$. As a convention, the allocations used in the optimal solution are denoted by $\vec S^*$.
For a sequence of requests $\requests$ and $i \in [n]$, denote by $\requests_{\leq i} = (w_1,\SS_1), \dots, (w_{i-1},\SS_{i-1}), (w_i, \SS_i)$ the subsequence of requests up to request~$i$. An online algorithm $\ALG$ is a family of functions $f_i : \requests_{\leq i}\to \SS_i$ mapping partial requests up to request~$i$ to a feasible allocation for request~$i$. For a sequence of requests $\mathcal{R}$, the cost of an online algorithm $\ALG$ with a family of functions $(f_i)_{i \in N}$ is the given by
$\ALG(\requests) = C(\vec S)$ where $\vec S = S_1 \times \dots \times S_n$ and $S_i = f_i(\requests_{\leq i})$.

We measure the performance of an online algorithm by its competitive ratio which is $\rho = \sup_{\requests} \ALG(\requests) / \OPT(\requests)$ where the supremum is taken over all finite sequences of requests for which $\OPT(\requests) > 0$. When the sequence of requests $\requests$ is clear from context, we write $\ALG$ and $\OPT$ instead of $\ALG(\requests)$ and $\OPT(\requests)$.

We analyze a very easy and natural online algorithm, which we call best reply algorithm. Let again denote $\vec{S}_{< i}$ the allocation vector of the algorithm before the $i$-th request is revealed. Then, $w_i c_r(w_r(\vec{S}_{< i}))$ is the per request cost at the arrival of request $i$ on resource $r$. Upon arrival of request~$i$, the best reply algorithm chooses an allocation $S_i \in \SS_i$ that minimizes the cost of that request, i.e.\ we choose some allocation $S_i$ such that,
\begin{align}
\label{eq:private-cost-greedy-constraint}
\sum_{r \in S_i}{w_i c_r(w_r(\vec{S}_{<i}) + w_i)} \leq \sum_{r \in S'_i}{w_i c_r(w_r(\vec{S}_{<i}) + w_i)}\;,
\end{align}
for all other feasible allocations $S'_i \in \SS_i$. This choice is motivated by best response moves in the corresponding congestion game. Note that the response steps used by the best reply algorithm are typically tractable and therefore the competitive ratio $\rho$ is also the approximation factor for the corresponding approximation algorithm.

\section{Unweighted Resource Allocation Problems} \label{sec:unweighted-nash}

In this section, we consider unweighted resource allocation problems with polynomial cost functions and derive an upper bound on the competitive ratio for the best reply algorithm. We give a general analysis and show a bound of $d(\Xi_d d)^{d+1} \in \mathcal{O}(((\Xi_d+\epsilon) d)^{d+1})$ for cost functions in $\mathcal{C}_d$. Here, $\mathcal{C}_d$ denotes the set of polynomials with non-negative coefficients and maximum degree $d$. $\Xi_d$ is the unique solution to the equation $d\left(2x e^{1/x} + x^2 - e^{2/x} - x^2 e^{1/x}\right) = e^{\frac{2}{x}}$ and $\lim_{d\rightarrow \infty} \Xi_d \approx 1.523$. This implies an exponential gap between the weighted and the unweighted case since for the weighted case the competitive ratio is about $(1.77d)^{d+1}$. 

Recall that the sequence of requests is given by $\requests = (1,\SS_1), \dots , (1,\SS_n)$ for unweighted resource allocation problems, i.e., $w_i = 1$ for all $i \in [n]$. This implies that the cost functions $c_r$ directly define the per request cost on this resource, since $w_ic_r = c_r$ in this case. The cost functions $c_r$ now only depend on the \emph{number} of requests that have chosen some resource $r$.

We use the definition of the algorithm in \eqref{eq:private-cost-greedy-constraint} to derive an optimization problem for all resources such that the solution to the problem relates to the competitive ratio of the algorithm. This approach is also known as the $(\lambda, \mu)$-smoothness framework. In this section, the constraints in the optimization problem only have to hold for integral $x=w_r(\vec S^*)$ and $y=w_r(\vec S)$.

To retain generality, let $d_r \leq d$ denote the maximal degree of the cost function $c_r$. We show that for every resource $r$, the smoothness condition is fulfilled for $\lambda_d =(\Xi_d d)^{d+1}$ and  $\mu_d=1-\frac{1}{d}$. Towards this end, we make a case distinction between (1) $x=0$, (2) $x\neq 0, y\leq \frac{1}{W(1.27)} d_r$, (3) $x=1, y>\frac{1}{W(1.27)}d_r$ and (4) $x\geq 2, y>\frac{1}{W(1.27)}d_r$.

\begin{theorem}\label{theorem:unweighted-nash-upper}
For any $d \in \mathbb{N}$, the best reply algorithm is $d(\Xi_d d)^{d+1}$ competitive for unweighted resource allocation problems with cost functions in $\mathcal{C}_d$. Here, $\Xi_d$ is the unique solution to \[d\left(2x e^{1/x} + x^2 - e^{2/x} - x^2 e^{1/x}\right)=e^\frac{2}{x}\] and $\Xi_d \leq \frac{1}{W\left(\frac{1.20d}{d+1}\right)}$.
\end{theorem}

\begin{proof}
First, note that the best reply algorithm is $4.24$ competitive for $d=1$ due to the work of Christodoulou, Mirrokni and Sidiropoulos \cite{ChristodoulouMS12}. We show in Lemma \ref{lem:unweighted-interval} that $\Xi_1 \geq \frac{1}{W(1.27/2)}\approx 2.39$, i.e.\ we can assume $d\geq 2$.

For the proof, we will assume wlog.\ that resource $r$ has cost equal to $x^{d_r}$ with $d_r\leq d$.  If this is not the case, we can achieve this setting by splitting up resources. Additionally, we assume that there are no resources $r$ with $w_r(\vec S) = 0$ and $w_r(\vec S^*) > 0$. If this is not the case, ignoring any contribution of these resources to $C(S^*)$ does only increase the competitive ratio.

The algorithm minimizes the current request's cost in each step, that is,
\[\sum_{r \in S_i}{\left(w_r(\vec S_{<i})+1\right)^{d_r}} \leq \sum_{r \in S^*_i}{\left(w_r(\vec S_{<i})+1\right)^{d_r}}.\]
The total cost can be written as the sum of the marginal increases to the total cost functions, i.e.\ we can also write
\begin{align*}
C(\vec S) &= \sum_{i \in [n]}\sum_{r \in S_i}{\left(\left(w_r(\vec S_{<i})+1\right)^{{d_r}+1}-w_r(\vec S_{<i})^{d_r+1}\right)}\\
&=\sum_{i \in [n]}\sum_{r \in S_i}\left(\sum_{k=0}^{d_r}w_r(\vec S_{<i})^k \binom{d_r+1}{k} + \left(w_r(\vec S_{<i})+1\right)^{d_r} \cdot (d_r+1)\right)\\
&\qquad - \sum_{i \in [n]}\sum_{r \in S_i}\left(\left(w_r(\vec S_{<i})+1\right)^{d_r}\cdot (d_r+1)\right)\\
&= \sum_{i \in [n]}\sum_{r \in S_i}\left((d_r+1)\left(w_r(\vec S_{<i})+1\right)^{d_r} + \sum_{k=0}^{d_r}{w_r(\vec S_{<i})^k\frac{(d_r+1)!}{k! (d_r+1-k)!}}\right)\\
&\qquad - \sum_{i \in [n]}\sum_{r \in S_i}\sum_{k=0}^{d_r}\left(w_r(\vec S_{<i})^k\frac{d_r!}{k!(d_r-k)!}\cdot (d_r+1)\right),
\end{align*}
where we used that $\sum_{k=0}^{d_r+1}{a^k\binom{d_r+1}{k}} = (a+1)^{d_r+1}$ in the first step.
We get
\begin{align*}
C(\vec S) &\leq \sum_{i \in [n]}{\sum_{r \in S^*_i}{(d_r+1)\left(w_r(\vec S_{<i})+1\right)^{d_r}}}\\
&\qquad- \sum_{i \in [n]}{\sum_{r \in S_i}{\sum_{k=0}^{d_r-1}{w_r(\vec S_{<i})^k\binom{d_r+1}{k}(d_r-k)}}}\;,
\end{align*}
by using the definition of the algorithm. In the following, we use $w_r(\vec S_{<i})\leq w_r(\vec S)$ and  that $w_r(\vec S_{<i})^k$ can be written as $j-1$ in the second sum, if $i$ is the $j$-th request that has been allocated to $r$ in $\vec S$. We get

\begin{align*}
C(\vec S) &\leq \sum_{r \in R}{(d_r\!+\!1)w_r(\vec S^*)(w_r(\vec S)\!+\!1)^{d_r}}\\
&\qquad- \sum_{r \in R}{\sum_{j=1}^{w_r(\vec S)}{\sum_{k=0}^{d_r-1}{(j-1)^k \binom{d_r\!+\!1}{k}(d_r\!-\!k)}}}
\end{align*}
Now we bound the inner sum of the triple sum with the integral, which leads to
\begin{align*}
&\sum_{r \in R}\sum_{j=1}^{w_r(\vec S)}{\sum_{k=0}^{d_r-1}{(j-1)^k \binom{d_r+1}{k}(d_r-k)}}\\
&\quad\geq \sum_{r \in R}{\sum_{k=0}^{d_r-1}{\binom{d_r+1}{k}(d_r-k)\int_{1}^{w_r(\vec S)}{(j-1)^k dj}}}\\
&\quad= \sum_{r \in R}{\sum_{k=0}^{d_r-1}{\binom{d_r+1}{k}(d_r-k)\frac{1}{k+1}\left(w_r(\vec S)-1\right)^{k+1}}}\\
&\quad= \sum_{r \in R}{\sum_{k=1}^{d_r}{\binom{d_r+1}{k}\frac{d_r-k+1}{d_r-k+2}\left(w_r(\vec S)-1\right)^{k}}}\;.
\end{align*}
At this point, we split the sum and apply a variant of the binomial theorem used above as well as $\sum_{k=0}^{d+1} (a-1)^k\binom{d+2}{k} = a^{d+2} - (a-1)^{d+2}$. We get
\begin{align*}
\sum_{r \in R}\sum_{k=1}^{d_r}&\binom{d_r+1}{k}\frac{d_r-k+1}{d_r-k+2}\left(w_r(\vec S)-1\right)^{k}\\
&= \sum_{r \in R}\sum_{k=0}^{d_r+1} \left(\left(w_r(\vec S)-1\right)^{k}\binom{d_r+1}{k}\frac{d_r-k+1}{d_r-k+2}\right) - \frac{d_r+1}{d_r+2}\\
&= \sum_{r \in R}\sum_{k=0}^{d_r+1} \left(\left(w_r(\vec S)-1\right)^{k}\binom{d_r+1}{k}\left(1-\frac{1}{d-k+2}\right)\right)-\frac{d_r+1}{d_r+2}\\
&= \sum_{r \in R} w_r(\vec S)^{d+1} - \sum_{k=0}^{d_r+1}\left(\left(w_r(\vec S) -1\right)^k\binom{d_r+2}{k}\frac{1}{d_r+2}\right) - \frac{d_r+1}{d_r+2}\\
&=\sum_{r\in R} w_r(\vec S)^{d+1} - \frac{(w_r(\vec S)^{d_r+2} - (w_r(\vec S) - 1)^{d_r+2})}{d_r+2} - \frac{d_r+1}{d_r+2}\;.
\end{align*}

In Proposition \ref{prop:unweightedNash}, we will show that choosing $\lambda_d =( \Xi_d d)^{d+1}$ and $\mu_d = 1-\frac{1}{d}$ fulfills the condition
\begin{align}
(d_r\!+\!1)(y\!+\!1)^{d_r}x - y^{d_r+1} + \frac{y^{d_r+2}-(y\!-\!1)^{d_r\!+\!2} + d_r\!+\!1}{d_r+2} \leq \lambda_d x^{d_r+1} + \mu_d y^{d_r+1}
\label{eq:lambdaMuUnweighted}
\end{align} for all $x \in \mathbb{N}_{\geq 0}, y \in \mathbb{N}_{\geq 1}$ and $d_r \leq d$. Note that we will only show the inequality for $y \geq 1$. However, omitting resources with $w_r(S)=0$ can only increase the approximation bound. Applying this to all $r \in R$ yields
\[C(\vec S) \leq \lambda_d C(\vec S^*) + \mu_d C(\vec S),\]
that is, a competitive ratio of $\frac{\lambda_d}{1-\mu_d}$. Thus, we seek to minimize $\frac{\lambda_d}{1-\mu_d}$ subject to the inequality \eqref{eq:lambdaMuUnweighted}. We will show in Proposition \ref{prop:unweightedNash} that there are $\lambda_d$ and $\mu_d$ such that $\frac{\lambda_d}{1-\mu_d}$ is upper bounded by $d(\Xi_d d)^{d+1}$. In Lemma~\ref{lem:unweighted-interval}, we show that $\frac{1}{W\left(\frac{1.27d-1}{d+1}\right)} \leq \Xi_d \leq \frac{1}{W\left(\frac{1.20d-1}{d+1}\right)}$. For a plot of the upper and lower bounds on $\Xi_d$, see Figure \ref{fig:XidPlot}. Note that a numerical analysis shows that $\lim_{d \rightarrow \infty}{\sqrt[d+1]{d} \Xi_d} \approx 1.523$, i.e.\ the value for $\Xi_d$ seems to be quite close to our lower bound for large $d$, since $\lim_{d \rightarrow \infty}{\frac{1}{W(\frac{1.27d-1}{d+1})}} \approx 1.520$.
\end{proof}

It remains to show that there are choices for $\lambda_d$ and $\mu_d$ that give the claimed competitive ratio. The proof of the following proposition relies on several technical lemmas which are proven in the appendix.


\begin{proposition}
\label{prop:unweightedNash}
For any $1 \leq d_r \leq d$, there are $\lambda_d$, $\mu_d$ with
\[\frac{\lambda_d}{1-\mu_d}\leq d(\Xi_d d)^{d+1},\]
and $(d_r+1)(y+1)^{d_r} x - y^{d_r+1} + \frac{y^{d_r+2}-(y-1)^{d_r+2}}{d_r+2} + \frac{d_r+1}{d_r+2} \leq \lambda_d x^{d_r+1} + \mu_d y^{d_r+1}$ for all $x \in \mathbb{N}_{\geq 0}, y \in \mathbb{N}_{\geq 1}$ and $d_r \leq d$, where $\Xi_d$ is the solution to the equation $d(2xe^{\frac{1}{x}} + x^2 - e^\frac{2}{x} - x^2 e^\frac{1}{x}) = e^\frac{2}{x}$.
\end{proposition}

\begin{proof}
Let $1\leq d_r \leq d$. In this proof, we will distinguish 4 cases. First, we show that the inequality holds for $x=0$. Then, we show the result for $x\neq 0$, $y\leq \frac{d_r}{W(1.27)}$. Third, we consider the case $x=1$, $y>\frac{d_r}{W(1.27)}$ and finally we finish the proof with the case $x\geq 2$, $y>\frac{d_r}{W(1.27)}$. For all 4 cases, we choose $\mu_d = 1-\frac{1}{d}$ and $\lambda_d = (\Xi_d d)^{d+1}$.

\paragraph*{Case 1: $x=0$.\newline}
In this case, $(d_r+1)(y+1)^{d_r}x = 0$. We show in Lemma \ref{lem:unweightedCaseOne} in the appendix that 
$\frac{y^{d_r+2}-(y-1)^{d_r+2}+d_r+1}{d_r+2} \leq y^{d_r+1}  \quad \forall y \in \mathbb{N}_{\geq 1}$, i.e.\ we get
\[(d_r+1)(y+1)^{d_r} x - y^{d_r+1} + \frac{y^{d_r+2}-(y-1)^{d_r+2}}{d_r+2} + \frac{d_r+1}{d_r+2} \leq 0\] 
for $x=0$, all $d_r \in \mathbb{N}_{\geq 1}$ and all $y \in \mathbb{N}_{\geq 1}$.
This finishes the proof of Case 1.

For Cases 2-4, we can assume $x\geq 1$. In order to show that the constraint is fulfilled for all $x,y \in \mathbb{N}_{\geq 1}$  for the choice $\mu_d = 1 - \frac{1}{d}$, $\lambda_d = (\Xi_d d)^{d+1}$, note that the constraint is equivalent to
\begin{align*}
&\max_{x,y \in \mathbb{N}_{\geq 1}}\left\{\frac{(d_r+1)(y+1)^{d_r} x - y^{d_r+1} + \frac{y^{d_r+2}-(y-1)^{d_r+2}+d_r+1}{d_r+2}- \frac{d-1}{d} y^{d_r+1}}{x^{d_r+1}}\right\}\\ &\qquad\leq (\Xi_d d)^{d+1}\;.
\end{align*}

\paragraph*{Case 2: $x \neq 0, y\leq \frac{d_r}{W(1.27)}$.\newline}
First, we will reconsider the inequality $\frac{y^{d_r+2}-(y-1)^{d_r+2}+d_r+1}{d_r+2} \leq  y^{d_r+1}$, which has been proven in Case 1. 
We get
\begin{align*}
&\max_{\substack{x,y\in \mathbb{N}_{\geq 1},\\y\leq \frac{d_r}{W(1.27)}}}\left\{\frac{(d_r+1)(y+1)^{d_r} x - y^{d_r+1} + \frac{y^{d_r+2}-(y-1)^{d_r+2} + d_r+1}{d_r+2}- \frac{d-1}{d}y^{d_r+1}}{x^{d_r+1}}\right\}\\
&\leq \max_{\substack{x,y\in \mathbb{N}_{\geq 1},\\y\leq \frac{d_r}{W(1.27)}}}\left\{\frac{(d_r+1)(y+1)^{d_r} x}{x^{d_r+1}}\right\}\leq \max_{\substack{x,y\in \mathbb{N}_{\geq 1},\\y\leq \frac{d_r}{W(1.27)}}}\left\{(d_r+1)(y+1)^{d_r}\right\}\\
&\leq \biggl(\frac{d_r}{W(1.27)}+1\biggr)^{d_r+1}\\
&\leq \biggl(\frac{d}{W(1.27)}+1\biggr)^{d+1} \leq (\Xi_d d)^{d+1},
\end{align*}
where the last inequality is shown in Lemma \ref{lem:Case2} in the appendix.

\paragraph*{Case 3: $x =1, y > \frac{d_r}{W(1.27)}$.\newline}

We start the proof of this case by plugging in $x=1$. This yields
\begin{align*}
M &:= \max_{y}\left\{(d_r\!+\!1)(y\!+\!1)^{d_r} - \left(2\!-\!\frac{1}{d}\right) y^{d_r+1} + \frac{y^{d_r+2}-(y\!-\!1)^{d_r+2} + d_r\!+\!1}{d_r+2}\right\}\\
&\leq\max_{y} \left\{y^{d_r} \left((d_r+1)\biggl(1+\frac{1}{y}\biggr)^{\!d_r} - \biggl(2-\frac{1}{d_r}\biggr) y + y^2\frac{1 - (1-\frac{1}{y})^{d_r+2}}{d_r+2}\right)\right\}\\
&\quad + \frac{d_r+1}{d_r+2}\;.
\end{align*}
We use that $(1+\frac{1}{y})^y \leq e \leq (1+\frac{1}{y})^{y+1}$ and get
\begin{align*}
M &\leq \max_{y}\left\{y^{d_r}\left((d_r+1)e^{d_r / y} - \biggl(2-\frac{1}{d}\biggr) y + y^2\frac{1 - e^{-\frac{d_r+2}{y-1}}}{d_r+2}\right)\right\} + \frac{d_r+1}{d_r+2}\;.
\end{align*}
In the following, we define $c = \frac{y}{d_r}$, replace $y$ with $c d_r$, and maximize over $c$ instead of $y$. We then obtain
\begin{align*}
M &< \max_{c}\left\{(cd_r)^{d_r}\left((d_r+1)e^\frac{1}{c} - \biggl(2\!-\!\frac{1}{d}\biggr) cd_r + c^2d_r\!\left(1\!-
 e^{-\frac{1}{c}}\right)\right)\right\}  + \frac{d_r\!+\!1}{d_r\!+\!2}.
\end{align*}
Here we use $\frac{1 - e^{-\frac{1}{c}}}{d_r} > \frac{1 - e^{-\frac{d_r+2}{cd_r-1}}}{d_r+2} \Leftrightarrow 2 > (d_r+2)e^{-\frac{1}{c}} - d_r e^{-\frac{d_r+2}{cd_r-1}}$. This is shown in Lemma \ref{lem:Case3TermSmallerTwo} in the appendix.

In the following, we will not explicitly calculate the maximizing $c$, but will derive an upper bound $\Xi_{d_r}$ on the maximum possible $c$, dependent on $d_r$. In order to do so, note that we show in Lemma \ref{lem:Case3TermDecreasing} in the appendix that the term $(d_r+1)(e)^\frac{1}{c} - (2-\frac{1}{d}) cd_r + c^2d_r(1 - e^{-\frac{1}{c}})$ is monotonically decreasing in $c$. Additionally, note that the whole expression is negative, if the term $(d_r+1)e^{1/c} - (2-\frac{1}{d}) cd_r + c^2d_r(1 - e^{-\frac{1}{c}})$ is negative. Thus, we conclude that the maximum is not attained for all $c$ such that 
\begin{align*}
d_r < -\frac{e^\frac{1}{c}}{e^\frac{1}{c} - (2-\frac{1}{d})c + c^2 - c^2e^{-\frac{1}{c}}} = \frac{e^\frac{2}{c}}{(2-\frac{1}{d})ce^\frac{1}{c} + c^2 - e^\frac{2}{c} - c^2e^\frac{1}{c}}\;.
\end{align*}
We conclude that $\Xi_{d_r}$ is defined such that
\[d_r = \frac{e^\frac{2}{\Xi_{d_r}}}{(2-\frac{1}{d})\Xi_{d_r}e^\frac{1}{\Xi_{d_r}} + \Xi_{d_r}^2 - e^\frac{2}{\Xi_{d_r}} - \Xi_{d_r}^2e^\frac{1}{\Xi_{d_r}}}\]
is an upper bound on the maximizing $c$.
Lemma \ref{lem:XidUnique} shows that $\Xi_{d_r}$ is unique. We conclude that $cd_r \leq \Xi_{d_r}d_r$, then we argue in Lemma \ref{lem:dXidIncreasing} in the appendix that $\Xi_{d_r}d_r$ is monotonically increasing in $d_r$ and thus, can be upper bounded by $\Xi_d d$. We use that the second part of the product is decreasing in $c$ to get 

\begin{align*}
&(\Xi_d d)^{d}\max_{c}\left\{\left((d_r+1)e^\frac{1}{c} - (2-\frac{1}{d}) cd_r + c^2d_r\left(1 - e^{-\frac{1}{c}}\right)\right)\right\} + \frac{d+1}{d+2}\\
& \leq (\Xi_d d)^{d} \left((d_r+1)e^{W(1.27)} - \frac{2d_r-1}{W(1.27)} + \frac{d_r\left(1 - e^{-W(1.27)}\right)}{W(1.27)^2}\right) + \frac{d+1}{d+2}\\
& \leq (\Xi_d d)^{d} \left(0.0042d + 3.46\right) + \frac{d+1}{d+2} \leq (\Xi_d d)^{d+1}\;,
\end{align*}
where the last inequality can be checked separately for $d=2$ and $d\geq 3$ with the help of the lower bound on $\Xi_d$ shown in Lemma \ref{lem:unweighted-interval}.
\paragraph*{Case 4: $x \geq 2, y > \frac{d_r}{W(1.27)}$.\newline}
We use the inequality $\frac{y^{d_r+2}-(y-1)^{d_r+2} + d_r + 1}{d_r+2} \leq y^{d_r+1}$ from Case 1 again and get
\begin{align*}
M&:= \max_{\substack{x\geq 2,\\y > \frac{d_r}{W(1.27)}}}\left\{\frac{(d_r\!+\!1)(y\!+\!1)^{d_r} x - y^{d_r+1} + \frac{y^{d_r+2}-(y-1)^{d_r+2}}{d_r+2}- \frac{d-1}{d} y^{d_r+1}}{x^{d_r+1}}\right\} \nonumber \\
&\leq \max_{x\geq 2,y > \frac{d_r}{W(1.27)}}\left\{\frac{(d_r+1)(y+1)^{d_r} x - (1-\frac{1}{d}) y^{d_r+1}}{x^{d_r+1}}\right\}\;.
\end{align*}
Again, we write $y = c \cdot d_r$ and optimize over all $c>\frac{1}{W(1.27)}$ instead. We write
\begin{align}
M &= \max_{x \geq 2 , c> \frac{1}{W(1.27)}} \left\{\frac{(d_r+1)(cd_r+1)^{d_r} x - (1-\frac{1}{d}) (cd_r)^{d_r+1}}{x^{d_r+1}}\right\} \nonumber \\
&\leq \max_{x \geq 2 , c> \frac{1}{W(1.27)}}\left\{\left(\frac{c}{x}d_r\right)^{d_r}\left((d_r+1)\left(1+\frac{1}{cd_r}\right)^{d_r} - \left(1-\frac{1}{d}\right)\frac{cd_r}{x}\right)\right\}\nonumber \\ \label{eq:case4}
&\leq \max_{x \geq 2 , c> \frac{1}{W(1.27)}}\left\{\left(\frac{c}{x}d_r\right)^{d_r}\left((d_r+1)\e^{\frac{1}{c}} - (d_r-1)\frac{c}{x}\right)\right\} \;.
\end{align}
We will proceed by deriving an upper bound on $\frac{c}{x}$ for the maximizing $c$ and $x$. In order to do so, set $z = \frac{c}{x}$ and consider the second part of the product.
\begin{align*}
&(d_r+1)\e^{\frac{1}{c}} - (d_r-1)\frac{c}{x} = (d_r+1)\e^{\frac{1}{xz}} - (d_r-1)z\leq (d_r+1)\e^{\frac{1}{2z}} - (d_r-1)z
\end{align*}
This term is monotonically decreasing in $z$, thus we get an upper bound on $z$, maximizing term \eqref{eq:case4} by setting
\begin{align*}
& (d_r+1)\e^{\frac{1}{2z}} - (d_r-1)z = 0\Leftrightarrow \frac{1}{2z} e^\frac{1}{2z} = \frac{d_r-1}{2(d_r+1)}\Leftrightarrow z = \frac{1}{2 W\left(\frac{d_r-1}{2(d_r+1)}\right)}\;,
\end{align*}
we assume without loss of generality in term \eqref{eq:case4} that $z \coloneqq \frac{c}{x} \leq \frac{1}{2 W\left(\frac{d_r-1}{2(d_r+1)}\right)}$. This leads to
\begin{align*}
M &\leq \max_{z \leq \frac{1}{2W\left(\frac{d_r-1}{2(d_r+1)}\right)}}\left\{\left(zd_r\right)^{d_r}\left((d_r+1)\e^{W(1.27)} - (d_r-1)z\right)\right\}\;.
\end{align*}
Lemma \ref{lem:zLemma} in the appendix shows that not only the maximum of \eqref{eq:case4} but also the maximum of the upper bound is attained at $z = \frac{1}{2W\left(\frac{d_r-1}{2(d_r+1)}\right)}$.
We conclude that we can upper bound the term by
\begin{align*}
&\left(\frac{d_r}{2W\left(\frac{d_r-1}{2(d_r+1)}\right)} \right)^{\!d_r}\left((d_r+1)\e^{W(1.27)} - (d_r-1) \frac{1}{2W\left(\frac{d_r-1}{2(d_r+1)}\right)} \right) \;.
\end{align*}
It remains to show that this is in fact upper bounded by $\left(\Xi_d d\right)^{d+1}$. This can be easily checked for $d=2,3,4,5$ by using $\Xi_d \geq \frac{1}{W\left(\frac{1.27d-1}{d+2}\right)}$ from Lemma \ref{lem:unweighted-interval}. For $d\geq 6$, the right side of the expression is bounded by
\begin{align*}
&\left(d_r\left(e^{W(1.27)} - \frac{1}{2W\left(\frac{d_r-1}{2(d_r+1)}\right)}\right) + e^{W(1.27)} + \frac{1}{2W\left(\frac{d_r-1}{2(d_r+1)}\right)}\right)\\
&\qquad\leq \left(d_r\left(e^{W(1.27)} - \frac{1}{2W\left(\frac{1}{2}\right)} + \frac{1}{6}\left( e^{W(1.27)} + \frac{1}{2W\left(\frac{1}{6}\right)}\right)\right)\right) \leq 1.41d_r\;.
\end{align*}
This gives us
\begin{align*}\left(\frac{d_r}{2W\left(\frac{d_r-1}{2(d_r+1)}\right)} \right)^{d_r}\left(1.41 d_r\right) \leq \left(\Xi_{d_r}d_r\right)^{d_r+1} \leq \left(\Xi_{d}d\right)^{d+1}\;.
\end{align*}
Here, we used that $\Xi_d \geq \frac{1}{W\left(1.27\right)} \approx 1.52$ (Lemma \ref{lem:unweighted-interval}) and that $\Xi_{d_r} d_r$ is increasing in $d_r$, see Lemma~\ref{lem:dXidIncreasing} in the appendix.
\end{proof}

\section{Weighted Resource Allocation Problems}\label{sec:weighted-nash-upper}

In this section, we revisit some upper bound on the competitive ratio of the best reply algorithm for weighted resource allocation problems with polynomial cost functions in $\mathcal{C}_d$. Awerbuch et al.~\cite{AwerbuchAGKKV95} have shown an upper bound of $\left(\Psi_d\right)^{d+1}$, where $\Psi_d$ is the solution to the equation $(d+1)(x+1)^d = x^{d+1}$ for singleton instances where the the cost of each resource is the identity. Bil{\`o} and Vinci~\cite{BiloV17} showed that the worst case for the competitive ratio is obtained for singletons but their proof crucially relies on non-identical cost functions. Bil{\`o} and Vinci also claimed that the $\Psi_d$ is the correct competitive ratio for arbitrary games, but their paper does not contain a proof of this result. For completeness, a proof of Theorem~\ref{thm:bilo} is contained in the appendix. 


\begin{theorem}[Bil{\`o} and Vinci~\cite{BiloV17}]\label{theorem:weighted-nash-upper}
\label{thm:bilo}
For polynomial costs in $\mathcal{C}_d$, the competitive ratio of the best reply algorithm is at most $\Psi_d^{d+1}$ where $\Psi_d$ is the unique solution to the equation $(d+1)(x+1)^d = x^{d+1}$.
\end{theorem}

The following theorem provides a closed expression that approximates $\Psi_{d}$ with small error. 

\begin{theorem}
\label{thm:psi}
The equation $(d+1)(x+1)^d = x^{d+1}$ has a unique solution $\Psi_d$ in $\RR_{\geq 0}$ for all $d \in \RR_{\geq 0}$. Moreover, $\Psi_d \in  [d/W(\frac{d}{d+1}) -1,d/W(\frac{d}{d+1})]$, where $W : \RR_{\geq 0} \to \RR_{\geq 0}$ is the Lambert-W function on $\RR_{\geq 0}$.   
\end{theorem}
\begin{proof}
We first show that the equation $(d+1)(x+1)^d = x^{d+1}$ has a unique solution. Since this solution is not $x=0$ we may assume $x \neq 0$, take logarithms and obtain the equivalent equation
$\log (d+1) + d\log (x+1) = (d+1) \log x$. 
Rearranging terms yields
\begin{align}
\label{eq:psi_log_equation}
\log(x+1) - \log(x)  = \frac{\log x - \log (d+1)}{d}\;.
\end{align}
The left hand side of this equation is decreasing in $x$ and takes values in $(0,\infty)$. The right hand side is increasing in $x$ and for $x \in [d+1,\infty)$, it takes values in $(0,\infty)$. Thus, the equation has a unique solution which we denote by $\Psi_d$.

To get an approximate closed form expression for $\Psi_d$, we use \eqref{eq:psi_log_equation} and the mean value theorem to obtain
$\frac{1}{\xi} = \frac{\log \Psi_d - \log (d+1)}{d}$ for some $\xi \in (\Psi_d,\Psi_d+1)$. 
We obtain
\begin{align*}
\Psi_d \in \Biggl\{ x \in [d+1,\infty) : \frac{1}{x+1} \leq \frac{\log x - \log (d+1)}{d} \leq \frac{1}{x} \Biggr\}.
\end{align*}
As $\log x$ is strictly increasing in $x$ and both $\frac{1}{x+1}$ and $\frac{1}{x}$ are decreasing, we obtain $\Psi_d \in [a,b]$ where $a$ is the unique solution to the equation $\frac{d}{x+1} = \log x - \log(d+1)$ and $b$ is the unique solution to the equation $\frac{d}{x} = \log x - \log(d+1)$.
The latter equation gives
\begin{align*}
\frac{d}{b} &= \log \frac{b}{d+1} & &\Leftrightarrow & e^{d/b} \frac{d}{b} &= \frac{d}{d+1}\;. 
\end{align*}
Using that $W$ is bijective on $\RR_{\geq 0}$ and that $W(xe^x) = x$ for all $x \in \RR_{\geq 0}$, this implies $\frac{d}{b} = W(\frac{d}{d+1})$ and, hence, $b = d /W(\frac{d}{d+1})$. To get a bound on $a$, note that $a \geq a'$ where $a'$ solves $\frac{d}{a'+1} = \log (a'+1) - \log (d+1)$. Substituting $b = a'+1$, we obtain $a' = d/W(\frac{d}{d+1}) - 1$ as before.
\end{proof}

\newpage
\bibliographystyle{plain}
\bibliography{literature}

\newpage
\appendix
\section{Appendix}

\subsection{Unweighted Resource Allocation Problems}

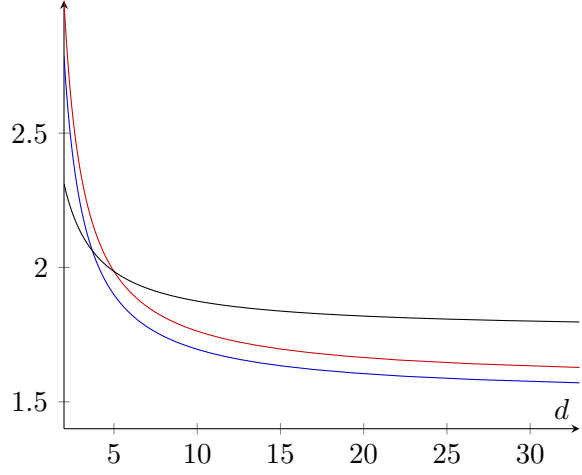
\begin{figure}
\centering
\begin{tikzpicture}
  \begin{axis}[
  	  ymin=1.4,
  	  xlabel=$d$,
      samples=100,
      enlarge y limits=false,
      axis lines=middle,
    ] 
    \addplot [red!80!black, domain=1.62801:2.99295] ({(1+(1/x * exp(1/x)))/(1.2-(1/x * exp(1/x)))}, {x});
        \addplot [blue!80!black, domain=1.57082:2.78838] ((1+(1/x * exp(1/x)))/(1.27-(1/x * exp(1/x))), x);
        \addplot [black, domain=1.79725:2.3118] (((1/x * exp(1/x)))/(1-(1/x * exp(1/x))), x);

  \end{axis}
\end{tikzpicture}
\caption{A plot of our upper and lower bounds on $\Xi_d$ (in red and blue) and the weighted coefficient $\frac{1}{W(\frac{d}{d+1})}$ (in black).}
\label{fig:XidPlot}
\end{figure}

\begin{lemma}\label{lem:XidUnique}
The equation
\[d_r\left((2-\frac{1}{d})x e^{\frac{1}{x}} + x^2 - e^{\frac{2}{x}} - x^2 e^\frac{1}{x}\right) = e^\frac{2}{x}\]
has a unique solution $\Xi_{d_r}$ in $\mathbb{R}_{\geq 0}$ for all $d_r \in \mathbb{R}_{\geq 2}$.
\end{lemma}
\begin{proof}
We start the proof by reformulating the equation. We observe
\begin{align*}
&\phantom{\Leftrightarrow} d_r\left((2-\frac{1}{d})x e^{\frac{1}{x}} + x^2 - e^{\frac{2}{x}} - x^2 e^\frac{1}{x}\right) = e^\frac{2}{x}\\
&\Leftrightarrow (2-\frac{1}{d})x + x^2 e^{-\frac{1}{x}} - e^{\frac{1}{x}} - x^2 = \frac{1}{d_r} e^\frac{1}{x}\\
&\Leftrightarrow x^2 e^{-\frac{1}{x}} - x^2 +(2-\frac{1}{d})x = e^\frac{1}{x}\left(1+\frac{1}{d_r}\right),
\end{align*}
where we used that $d\neq 0$ and $e^{\frac{1}{x}}\neq 0$. Note that $e^\frac{1}{x}(1+\frac{1}{d_r})$ is monotonically decreasing in $x$. We will finish the proof by showing that $x^2 e^{-\frac{1}{x}} - x^2 +(2-\frac{1}{d})x$ is monotonically increasing in $x$. In order to do so, we consider the derivative
\begin{align*}
&\frac{\partial}{\partial x}\left(x^2 e^{-\frac{1}{x}} - x^2 +(2-\frac{1}{d})x\right)\\
&=(2x+1)e^{-\frac{1}{x}}-2x + 2-\frac{1}{d} > (2x+1)e^{-\frac{1}{x}}-2x + 1.5\;.
\end{align*}
In the following, we will distinguish the cases $x<\frac{3}{4}$, $\frac{3}{4}\leq x \leq 1.18$, $1.18 < x \leq 2$, and $x > 2$ and show that in all cases $(2x+1)e^{-\frac{1}{x}}-2x + 1.5 > 0$.

\paragraph{Case 1: $x<\frac{3}{4}$.}
\begin{align*}
(2x+1)e^{-\frac{1}{x}}-2x + 1.5 > -2x + 1.5 \geq 0\;.
\end{align*}

\paragraph{Case 2: $\frac{3}{4}\leq x \leq 1.18$.}
\begin{align*}
&(2x+1)e^{-\frac{1}{x}}-2x + 1.5 > (2x+1)e^{-\frac{4}{3}}-2x + 1.5 > -1.48x + 1.76 > 0\;.
\end{align*}

\paragraph{Case 3: $1.18 < x \leq 2$.}
\begin{align*}
(2x+1)e^{-\frac{1}{x}}-2x + 1.5 &> (2x+1)e^{-\frac{1}{1.18}}-2x + 1.5\\
&> -1.144x + 2.356 > 0\;.
\end{align*}

\paragraph{Case 4: $2 < x$.}
\begin{align*}
&(2x+1)e^{-\frac{1}{x}}-2x + 1.5 > (2x+1)(1-\frac{1}{x})-2x + 1.5 = (1-\frac{1}{x}) -\frac{1}{2} > 0\;,
\end{align*}
which completes the proof.
\end{proof}

\begin{lemma}\label{lem:unweighted-interval}
Let $\Xi_d$ denote the unique solution to the equation
\[d_r\left((2-\frac{1}{d})x e^{\frac{1}{x}} + x^2 - e^{\frac{2}{x}} - x^2 e^\frac{1}{x}\right) = e^\frac{2}{x}\;.\]
Then,
\[\frac{1}{W\left(\frac{1.27d}{d+1}\right)} \leq \Xi_d \leq \frac{1}{W\left(\frac{1.20d}{d+1}\right)}\;,\]
where $W : \mathbb{R}_{\geq 0} \rightarrow \mathbb{R}_{\geq 0}$ is the Lambert-W function on $\mathbb{R}_{\geq 0}$.
\end{lemma}

\begin{proof}
First, we reformulate the equation as in Lemma \ref{lem:XidUnique}.
\begin{align*}
& & d_r\left((2-\frac{1}{d})x e^{\frac{1}{x}} + x^2 - e^{\frac{2}{x}} - x^2 e^\frac{1}{x}\right) &= e^\frac{2}{x} \\
&\Leftrightarrow & x^2 e^{-\frac{1}{x}} - x^2 +(2-\frac{1}{d})x &= e^\frac{1}{x}\biggl(1+\frac{1}{d_r}\biggr).
\end{align*}
We proceed by showing upper and lower bounds on $\Xi_d$. Note that by using $x\neq 0$, the equation above is equivalent to
\[\frac{1}{x}e^{\frac{1}{x}} = \frac{d}{d+1}\left(xe^{-\frac{1}{x}}-x+2-\frac{1}{d}\right) = \frac{d}{d+1}\left(xe^{-\frac{1}{x}}-x+2\right) -\frac{1}{d+1}.\]
In order to simplify notation, let us define $g(x) = xe^{-\frac{1}{x}}-x+2$. We will proceed to derive bounds on $\frac{1}{x}e^{\frac{1}{x}}$ by showing upper and lower bounds on $g$. First, we argue that $g$ is monotonically decreasing. Note that $\frac{\partial}{\partial x}g(x) = e^{-\frac{1}{x}}\left(1+\frac{1}{x}\right)-1 = e^{-\frac{1}{x}}\left(1+\frac{1}{x}-e^\frac{1}{x}\right)$ is negative. Since
\[\lim_{x \rightarrow \infty}g(x) = 2+ \lim_{x \rightarrow \infty}{\frac{e^{-\frac{1}{x}}-1}{\frac{1}{x}}} = 2+ \lim_{x \rightarrow \infty}{\frac{\frac{1}{x^2}e^{-\frac{1}{x}}}{-\frac{1}{x^2}}} = 2+ \lim_{x \rightarrow \infty}{-e^{-\frac{1}{x}}} = 1,\]
we get that $1\leq g(x)\leq 2$ for all $x \in \mathbb{R}_{\geq 0}$. This property of $g(x)$ leads to the inequality
\[\frac{d}{d+1} - \frac{1}{d+1} \leq \frac{1}{x} e^\frac{1}{x} \leq \frac{2d}{d+1}  - \frac{1}{d+1}.\] By using the monotonicity of the Lambert-W function, we obtain
\begin{align*}
&\phantom{\Leftrightarrow} W\left(\frac{d-1}{d+1}\right) \leq \frac{1}{x} \leq W\left(\frac{2d-1}{d+1}\right)\\
&\Leftrightarrow \frac{1}{W\left(\frac{2-1}{2+1}\right)} \geq \frac{1}{W\left(\frac{d-1}{d+1}\right)} \geq x \geq \frac{1}{W\left(\frac{2d-1}{d+1}\right)} \geq \frac{1}{W\left(2\right)}.
\end{align*}
Using this upper and lower bound on $x$ for $g(x)$ again, we conclude
$1.11 \leq g(x) \leq 1.33$ for all $x \in \left[\frac{1}{W\left(2\right)}, \frac{1}{W\left(\frac{1}{3}\right)}\right]$, i.e.\ we get the following improved bounds on $\frac{1}{x}e^{\frac{1}{x}}$ and $x$. 
\begin{align*}
&\phantom{\Leftrightarrow} \frac{1.11d-1}{d+1} \leq \frac{1}{x} e^\frac{1}{x} \leq \frac{1.33d-1}{d+1}\\
&\Leftrightarrow W\left(\frac{1.11d-1}{d+1}\right) \leq \frac{1}{x} \leq W\left(\frac{1.33d-1}{d+1}\right)\\
&\Leftrightarrow 2.16 \geq \frac{1}{W\left(\frac{1.11 \cdot 2-1}{2+1}\right)} \geq \frac{1}{W\left(\frac{1.11 d-1}{d+1}\right)}\\
&\phantom{\Leftrightarrow 2.16\ } \geq x \geq \frac{1}{W\left(\frac{1.33 d-1}{d+1}\right)} \geq \frac{1}{W\left(1.33 \right)} \geq 1.47.
\end{align*}
Applying these new bounds on $g(x)$ yields $1.19 < g(x) < 1.28$, i.e.\ we can again conclude
\begin{align*}
W\left(\frac{1.19d-1}{d+1}\right)\leq \frac{1}{x} \leq W\left(\frac{1.28d-1}{d+1}\right),
\end{align*}
which leads to $1.51 < x < 2.05$ and $1.20 < g(x) < 1.27$ for all relevant candidates for $x$. Using this, we get the desired result of
\[\frac{1}{W\left(\frac{1.27d-1}{d+1}\right)} \leq x \leq \frac{1}{W\left(\frac{1.20d-1}{d+1}\right)},\]
which finishes the proof.
\end{proof}

\begin{lemma}
\label{lem:unweightedCaseOne}
For all $d \geq 1$ and all $y\geq 1$, the following inequality holds
\[\frac{y^{d+2} - (y-1)^{d+2} + d + 1}{d+2} \leq y^{d+1}\;.\]
\end{lemma}

\begin{proof}
We prove the lemma by induction. Let $d = 1$, then we have
\begin{align*}
\frac{y^{d+2} - (y-1)^{d+2} + d + 1}{d+2} &= \frac{y^3 - y^3 + 3y^2 - 3y + 1 + 2}{3}\\
&= y^2 - y + 1 \leq y^2\;,
\end{align*}
which holds true for all $y \geq 1$.

Next, assume the inequality holds for all $d = \{1, \dots, d'\}$, then for $d = d'+1$ and $y=1$ we have
\[\frac{y^{d'+3} - (y-1)^{d'+3} + d'+2}{d'+3} = \frac{1 +d'+2}{d'+3} = 1 \leq y^{d'+2}\;.\]
Furthermore, the difference of the left- and right-hand side of the inequality $f_{d'+1}(y) = \frac{y^{d'+3} - (y-1)^{d'+3} + d'+2}{d'+3} - y^{d'+2}$ is continuous and the derivative is
\begin{align*}f'_{d'+1}(y) &= \frac{(d'+3)y^{d'+2} - (d'+3)(y-1)^{d'+2}}{d'+3}+(d'+2)y^{d'+1}\\
&\leq y^{d'+2} - (y-1)^{d'+2} + d'+1 + (d'+2)y^{d'+1}\\
&= (d'+2)\left(\frac{y^{d'+2}-(y-1)^{d'+2}+d'+1}{d'+2} + y^{d'+1}\right) \leq 0\;.
\end{align*}
Where the last inequality holds by induction.
\end{proof}

\begin{lemma}
For all $d\geq 2$ it holds that
\[\frac{1}{W(1.27)}+\frac{1}{d} \leq \Xi_d,\]
where $\Xi_d$ is the solution to the equation $d(2xe^{\frac{1}{x}} + x^2 - e^\frac{2}{x} - x^2 e^\frac{1}{x}) = e^\frac{2}{x}$.  
\label{lem:Case2}
\end{lemma}

\begin{proof}
We start the proof by denoting $c = \frac{1}{W(1.27)} > 1.52$, and by using $\Xi_d \geq \frac{1}{W\left(\frac{1.27d-1}{d+1}\right)}$ which we show in Lemma \ref{lem:unweighted-interval}. Thus it suffices to show
\begin{align*}
& c+\frac{1}{d} \leq \frac{1}{W\left(\frac{1.27d-1}{d+1}\right)}\\
&\Leftrightarrow \frac{d}{cd+1} \geq W\left(\frac{1.27d-1}{d+1}\right)\\
&\Leftrightarrow \frac{d}{cd+1} e^{\frac{d}{cd+1}} \geq \frac{1.27d-1}{d+1}\\
&\Leftrightarrow \frac{1.27d-1}{d+1} - \frac{d}{cd+1} e^{\frac{d}{cd+1}} \leq 0.
\end{align*}
For $g(d)\coloneqq \frac{1.27d-1}{d+1} - \frac{d}{cd+1} e^{\frac{d}{cd+1}}$, we will prove that $g(d)$ is monotonically increasing in $d$. Together with the observation $\lim_{d \rightarrow \infty}{g(d)} = 0$, this completes the proof. In order to show that $g(d)$ is increasing, we consider
\begin{align*}
\frac{\partial}{\partial d}g(d) &= \frac{2.27}{(d+1)^2} - \frac{cd+d+1}{(cd+1)^3}e^{\frac{d}{cd+1}}\\
&\geq \frac{1}{(d+1)^2}\left(2.27 - \frac{cd+d+1}{c^2(cd+1)}e^\frac{1}{c}\right)\\
&\geq \frac{1}{(d+1)^2}\left(2.27-\left(1 + \frac{d}{cd+1}\right)\frac{1}{c^2}e^\frac{1}{c}\right)\\
&\geq \frac{1}{(d+1)^2}\left(2.27-\left(1 + \frac{1}{c}\right)\frac{1}{c^2}e^\frac{1}{c}\right) > \frac{0.88}{(d+1)^2}\;,
\end{align*}
which finishes the proof.
\end{proof}

\begin{lemma}
\label{lem:Case3TermSmallerTwo}
The inequality
\[(d_r+2)e^{-\frac{1}{c}} - d_r e^{-\frac{d_r+2}{cd_r-1}} < 2\]
holds for all $d_r \geq 2, c \geq 1$.
\end{lemma}

\begin{proof}
Fix any $d_r \in \mathbb{R}_+$ and define $g: \mathbb{R} \rightarrow \mathbb{R}$ with
\[c \mapsto (d_r+2)e^{-\frac{1}{c}} - d_r e^{-\frac{d_r+2}{cd_r-1}}.\] We will show the statement by showing that
$\lim_{c \rightarrow \infty} g(c) = 2$ and $g(c)$ is increasing in $c$. Observe that
\begin{align*}
\lim_{c \rightarrow \infty} (d_r+2)e^{-\frac{1}{c}} - d_r e^{-\frac{d_r+2}{cd_r-1}} = (d_r+2) - d_r = 2.
\end{align*}
In order to show that $g(c)$ is monotonically increasing in $c$, recall that $e^x < \frac{1}{1-x}$ for all $x<1$, thus we calculate
\begin{align*}
\frac{\partial}{\partial c}\left(g(c)\right) &= \frac{d_r+2}{c^2}e^{-\frac{1}{c}} - \frac{(d_r+2)d_r^2}{(cd_r-1)^2}e^{-\frac{d_r+2}{cd_r -1}}\\
&= \frac{d_r+2}{c^2}e^{-\frac{1}{c}}\left(1-\frac{c^2d_r^2}{(cd_r-1)^2} e^{-\frac{d_r+2}{cd_r -1} + \frac{1}{c}}\right)\\
&= \frac{d_r+2}{c^2}e^{-\frac{1}{c}}\left(1-\frac{c^2d_r^2}{(cd_r-1)^2}e^{\frac{-2c-1}{(cd_r -1)c} }\right)\\
&> \frac{d_r+2}{c^2}e^{-\frac{1}{c}}\left(1-\frac{c^2d_r^2}{(cd_r-1)^2}\frac{1}{\frac{(cd_r-1)c+2c+1}{(cd_r -1)c} }\right)\\
&=\frac{d_r+2}{c^2}e^{-\frac{1}{c}}\left(1-\frac{c^3d_r^2}{(cd_r-1)(c^2d_r-c+2c+1)}\right)\\
&=\frac{d_r+2}{c^2}e^{-\frac{1}{c}}\left(1-\frac{c^3d_r^2}{c^3d_r^2 +cd_r-c-1}\right) > 0 ,
\end{align*}
since $\frac{d_r+2}{c^2}e^{-\frac{1}{c}} > 0$, $d_r \geq 2$, and $c \geq 1$.
\end{proof}

\begin{lemma}
\label{lem:Case3TermDecreasing}
The term \[(d_r + 1)e^\frac{1}{c} - \biggl(2-\frac{1}{d}\biggr)c d_r +c^2d_r(1-e^{-\frac{1}{c}})\]
is monotonically decreasing in $c$.
\end{lemma}
\begin{proof}
Note that the term is differentiable and
\begin{align*}
D &:= \frac{\partial}{\partial c} \left((d_r + 1)e^\frac{1}{c} - \biggl(2-\frac{1}{d}\biggr)c d_r+c^2d_r (1-e^{-\frac{1}{c}})\right)\\
&=-\frac{d_r+1}{c^2}e^\frac{1}{c} - \biggl(2-\frac{1}{d}\biggr)d_r + 2cd_r + (-2cd_r-d_r)e^{-\frac{1}{c}}\\
&< -\frac{d_r+1}{c^2}(1+\frac{1}{c}) - 2d_r + 1 + 2cd_r - (2cd_r+d_r)\biggl(1-\frac{1}{c}\biggr)\\
&< -\frac{d_r+1}{c^2} -\frac{d_r+1}{c^3} + 1 - d_r + \frac{d_r}{c}\;.
\end{align*}
For $d_r=2$ we get
\begin{align*}
D &= -\frac{3}{c^2} -\frac{3}{c^3} - 1 + \frac{2}{c} = \frac{-c^3+2c^2-3c-3}{c^3}\;.
\end{align*}
where the numerator can easily be checked to be negative for all $c>0$. Thus, we can assume $d_r \geq 3$. In this case we calculate
\begin{align*}
D < 1 - d_r + \frac{d_r}{c} < d_r(\frac{1}{d_r}-1+W(1.27)) < d_r(\frac{1}{3}-1+0.66) <0\;,
\end{align*}
which completes the proof.
\end{proof}

\begin{lemma}
\label{lem:technicalDanielCase3}
Let $\Xi_d$ be the solution to the equation
\[d\left(2x e^{\frac{1}{x}} + x^2 - e^{\frac{2}{x}} - x^2 e^\frac{1}{x}\right) = e^\frac{2}{x}\;.\]
Then, $\Xi_d$ is decreasing in $d$.
\end{lemma}
\begin{proof}
The equation in the statement of the lemma is equivalent to
\[d\left(2x e^{-\frac{1}{x}} + x^2e^{-\frac{2}{x}} - 1 - x^2 e^{-\frac{1}{x}}\right) = 1.\]
The solution $x=\Xi_d$ to this equation is decreasing in $d$ if the expression in the bracket is increasing in $x$. We consider the derivative
\begin{align*}
&\frac{\partial}{\partial x} 2x e^{-\frac{1}{x}} + x^2e^{-\frac{2}{x}} - 1 - x^2 e^{-\frac{1}{x}}\\
&= 2e^{-\frac{1}{x}} + \frac{2}{x}e^{-\frac{1}{x}} + 2xe^{-\frac{2}{x}} + 2e^{-\frac{2}{x}} - 2xe^{-\frac{1}{x}} - e^{-\frac{1}{x}}\\
&= \frac{e^\frac{1}{x}(x+2-2x^2) + 2x(x+1)}{e^\frac{2}{x}x}.
\end{align*}
Note that Lemma~\ref{lem:unweighted-interval} guarantees that $\Xi_d \geq \frac{1}{W(1.27)} \approx 1.522$ for all $d$. The denominator is positive for $x\geq 1$, so the derivative is positive if the numerator is positive. We use that $\frac{3}{2} \leq x\leq \frac{1}{W\left(\frac{1.2d-1}{d+1}\right)} \leq \frac{1}{W\left(\frac{1.4}{3}\right)} \leq 2$ for all $d\geq 2$ because the Lambert $W$-function is monotonically increasing for positive arguments. We have
\begin{align*}
e^\frac{1}{x}(x+2-2x^2) + 2x(x+1) &\geq e^\frac{2}{3}(2x+2-2x^2)+2x(x+1)\\
&\geq 2(2x+2-2x^2)+2x(x+1)\\
&= -2x^2+6x+2 \geq 0,
\end{align*}
for all $\frac{3}{2} \leq x\leq 2$.
\end{proof}

\begin{lemma}
\label{lem:dXidIncreasing}
Let $\Xi_d$ be the (unique) solution to the equation
\[d\left(2x e^{\frac{1}{x}} + x^2 - e^{\frac{2}{x}} - x^2 e^\frac{1}{x}\right) = e^\frac{2}{x}.\]
Then, $\Xi_d d$ is monotonically increasing in $d$.
\end{lemma}

\begin{proof}
Let us define $\Xi_{\infty} \coloneqq \lim_{d \rightarrow \infty}{\Xi_d}$. Note that the limit exists due to Lemma \ref{lem:unweighted-interval} and Lemma \ref{lem:technicalDanielCase3}. Let us define $f: (\Xi_\infty, \infty) \rightarrow \mathbb{R_+}$ with
\[x \mapsto \frac{e^\frac{2}{x}}{2xe^{\frac{1}{x}} + x^2 - e^{\frac{2}{x}} - x^2e^\frac{1}{x}},\]
and $g: (\Xi_\infty, \infty) \rightarrow \mathbb{R_+}$, $x \mapsto xf(x)$.
Note that $f(\Xi_d) = d$, i.e.\
\[ \Xi_d d = \Xi_d f(\Xi_d) = g(\Xi_d). \]
Observe that, according to Lemma \ref{lem:technicalDanielCase3}, the solution $\Xi_d$ is monotonically decreasing in $d$. In the remainder of this proof, we show that $g(x)$ is monotonically decreasing in $x$ and therefore $g(\Xi_d)$ is increasing in $d$, which completes the proof.
First, reformulate the term
\[g(x) = \frac{x e^{\frac{2}{x}}}{2xe^{\frac{1}{x}} + x^2 - e^{\frac{2}{x}}- x^2 e^{\frac{1}{x}}} = \frac{e^{\frac{1}{x}}}{2 + xe^{-\frac{1}{x}} - \frac{1}{x}e^{\frac{1}{x}}- x}.\]
Note that $e^\frac{1}{x}$ is decreasing and positive. We will prove the lemma by showing that the denominator is monotonically increasing in $x$ and positive. We observe
\begin{align*}
&2 + \Xi_\infty e^{-\frac{1}{\Xi_\infty}} - \frac{1}{\Xi_\infty}e^{\frac{1}{\Xi_\infty}}- \Xi_\infty\\
&\quad= \left(\frac{1}{\Xi_\infty} e^{-\frac{1}{\Xi_\infty}}\right)\left(2\Xi_\infty e^\frac{1}{\Xi_\infty} + \left(\Xi_\infty\right)^2 - e^{\frac{2}{\Xi_\infty}}- \left(\Xi_\infty\right)^2 e^{\frac{1}{\Xi_\infty}}\right) = 0,
\end{align*}
so positivity of the denominator follows directly if we show that $2 + xe^{-\frac{1}{x}} - \frac{1}{x}e^{\frac{1}{x}}- x$ is strongly monotonically increasing for all $x \in (\Xi_d, \infty)$. In order to prove this, we consider
\begin{align*}
\frac{\partial}{\partial x}&\left(2 + xe^{-\frac{1}{x}} - \frac{1}{x}e^{\frac{1}{x}}- x\right) = -1 + e^{-\frac{1}{x}} + \frac{1}{x^3} e^{\frac{1}{x}} + \frac{1}{x^2} e^{\frac{1}{x}} + \frac{1}{x} e^{-\frac{1}{x}}\\
& > -1 + e^{-\frac{1}{x}} + \frac{1}{x^2} e^{\frac{1}{x}} + \frac{1}{x} e^{-\frac{1}{x}} = -1 + e^{-\frac{1}{x}}\left(1+\frac{1}{x}\right) + \frac{1}{x^2}e^{\frac{1}{x}}.
\end{align*}
In the following, we use that $e > \left(1+\frac{1}{x}\right)^x$, i.e.\
$e^\frac{1}{x} > 1+\frac{1}{x}$, and $e < \left(1+\frac{1}{x-1}\right)^{x}$, i.e.\ $e^{-\frac{1}{x}} > \left(1+\frac{1}{x-1}\right)^{-1} = \frac{x-1}{x}$ for all $x \in (\Xi_\infty, \infty)$. We get
\begin{align*}
-1 + &e^{-\frac{1}{x}}\left(1+\frac{1}{x}\right) + \frac{1}{x^2}e^{\frac{1}{x}} > -1 + \left(\frac{x-1}{x}\right)\left(1+\frac{1}{x}\right) + \frac{1}{x^2}\left(1+\frac{1}{x}\right)\\
&= -1 + \frac{x-1}{x}\frac{x+1}{x} + \frac{1}{x^2}+\frac{1}{x^3} = -1 + \frac{x^3-x}{x^3} + \frac{x}{x^3}+\frac{1}{x^3} > 0,
\end{align*}
i.e.\ the denominator of the term is strongly monotonically increasing in for all $x \in (\Xi_\infty, \infty)$ which completes the proof.
\end{proof}

\begin{lemma}
\label{lem:zLemma}
The maximum of $\left(zd_r\right)^{d_r}\left((d_r+1)\e^{W(1.27)} - (d_r-1)z\right)$, for $0 \leq z \leq \frac{1}{2W\left(\frac{d_r-1}{2(d_r+1)}\right)}$ is attained at $z = \frac{1}{2W\left(\frac{d_r-1}{2(d_r+1)}\right)}$.
\end{lemma}
\begin{proof}
In order to simplify notation, define $a \coloneqq  e^{W(1.27)}(d_r+1)d_r^{d_r}$, and $b \coloneqq (d_r-1)d_r^{d_r}$. Note that
\begin{align*}
\frac{\partial}{\partial z} &\left(e^{W(1.27)}(d_r+1)d_r^{d_r} z^{d_r} - (d_r-1)d_r^{d_r}z^{d_r+1}\right) = \frac{\partial}{\partial z} \left( a z^{d_r} - b z^{d_r+1}\right)\\
&= z^{d_r-1}\left(d_r a - (d_r+1)bz\right)\;.
\end{align*}
The derivative is positive for all 
\begin{align*}
z<\frac{d_ra}{(d_r+1)b} = \frac{d_r e^{W(1.27)}(d_r+1)d_r^{d_r} }{(d_r+1)(d_r-1)d_r^{d_r}} = \frac{d_r}{d_r-1}e^{W(1.27)}\;.
\end{align*}
It remains to show that $\frac{d_r}{d_r-1}e^{W(1.27)} \geq \frac{1}{2W\left(\frac{d_r-1}{2(d_r+1)}\right)}$.
This can easily be verified for $d_r \in \{2,3,4,5\}$. For $d_r \geq 6$ we have
\begin{align*}
\frac{d_r}{d_r-1}e^{W(1.27)} \geq e^{W(1.27)} > 1.93 > 1.84 > \frac{1}{2 W\left(\frac{5}{14}\right)} \geq \frac{1}{2W\left(\frac{d_r-1}{2(d_r+1)}\right)}\;,
\end{align*}
which finishes the proof.
\end{proof}

\subsection{Appendix: Weighted Resource Allocation Problems}\label{app:weighted}

\begin{lemma}
\label{lemma:optimization-weighted-nash}
Let $\C$ be a set of semi-convex and non-decreasing cost functions and let $\beta = \sup_{c \in \C, x \in \RR_{\geq 0}} \frac{c(x) + c'(x)x}{c(x)}$. Further, let $\lambda > 0$ and $\mu \in [0,\frac{1}{\beta})$ be such that
\begin{align*}
x c(x+y) \leq \lambda xc(x) + \mu yc(y)	\text{ for all $x,y \in \RR_{\geq 0}$ and $c \in \C$\;.}
\end{align*}
Then the best reply algorithm is $\frac{\beta\lambda}{1-\beta\mu}$-competitive.
\end{lemma}

\begin{proof}
Let $\vec S_{<i} = (f_1(\requests_{\leq 1}),\dots,f_{i-1}(\requests_{i-1}))$ denote the allocation vector before the $i$-th request is revealed and $c_r(w_r(\vec S_{<i}))$ denote the correspond cost on resource~$r$. Denoting the choice of the algorithm by $S_i = f_i(\requests_{\leq i})$ and an optimal allocation by $\vec S^*$, we obtain the inequality
\begin{align}
\sum_{r \in S_i} c_r(w_r(\vec S_{<i})+w_i) \leq \sum_{r \in S_i^*} c_r(w_r(\vec S_{<i}) + w_i)\;.	
\end{align}

The total cost after the $i$-th request is revealed is
\begin{align*}
C(\vec S_{\leq i}) &= \sum_{r \in R} w_r(\vec S_{\leq i})c_r(w_r(\vec S_{\leq i}))\\
&= \sum_{r \in S_i} (w_r(\vec S_{\leq i-1}) + w_i) c_r(w_r(\vec S_{\leq i-1}) + w_i)\\
&\quad + \sum_{r \notin S_i} w_r(\vec S_{\leq i-1})c_r(w_r(\vec S_{\leq i-1}))\\
&= C(\vec S_{\leq i-1}) + \sum_{r \in S_i} \Bigl((w_r(\vec S_{\leq i-1}) + w_i)c_r(w_r(\vec S_{\leq i-1}) + w_i) \\
&\quad -w_r(\vec S_{\leq i-1})c_r(w_r(\vec S_{\leq i-1}))\Bigr)\\
&\leq C(\vec S_{\leq i-1}) + \sum_{r \in S_i} w_i \Bigl( c_r(w_r(\vec S_{\leq i})) + c_r'(w_r(\vec S_{\leq i}))w_r(\vec S_{\leq i}) \Bigr) \\
&\leq C(\vec S_{\leq i-1}) + \beta \sum_{r \in S_i} w_i\, c_r(w_r(\vec S_{< i})+w_i) \leq C(\vec S_{\leq i-1})\\
&\quad + \beta\sum_{r \in S^*_i} w_i\, c_r(w_r(\vec S_{< i})+w_i)\;.
\end{align*}
Using this inequality $n$ times we obtain
\begin{align*}
C(\vec S) &= \sum_{i \in N} C(\vec S_{\leq i}) - C(\vec S_{\leq i-1}) \leq \sum_{i\in N}  \sum_{r\in S^*_i} \beta\, w_i\,c_r(w_r(\vec S_{<i}) + w_i)\\
&\leq \sum_{i\in N}  \sum_{r\in S^*_i} \beta\, w_i\,c_r(w_r(\vec S) + w_r(\vec S^*)) = \beta \sum_{r \in R} w_r(\vec S^*)c_r(w_r(\vec S) + w_r(\vec S^*))\\
&\leq \beta \sum_{r \in R} \lambda w_r(\vec S^*)c_r(w_r(\vec S^*)) + \mu w_r(\vec S)c_r(w_r(\vec S)) = \beta \lambda C(\vec S^*) + \beta\mu C(\vec S).
\end{align*}
Rearranging terms gives the claimed result.
\end{proof}

We proceed to use Lemma~\ref{lemma:optimization-weighted-nash} in order to give an upper bound on the competitive ratio of the best reply algorithm. We will express the competitive ratio in terms of the unique solution to the equation $(d+1)(x+1)^d = x^{d+1}$ which we will denote by $\Psi_d$.

\begin{proof}[of Theorem~\ref{theorem:weighted-nash-upper}]
By splitting up resources with cost functions $c$ into several resources, we may assume that each resource has a cost function of the form $c_r(x) = a_r x^{d_r}$ with $d_r \in [0,d]$ and $a_r \in \RR_{\geq 0}$. We then obtain
\begin{align*}
\beta = \sup_{c \in \C} \sup_{x \in \RR_{\geq 0}} \frac{c(x) + xc'(x)}{c(x)} = \sup_{a \in \RR_{\geq 0}, d_r \in [0,d]} \sup_{x \in \RR_{\geq 0}} \frac{a x^{d_r} + a {d_r}x^{d_r}}{a x^{d_r}} = d+1\;. 	
\end{align*}

In light of Lemma~\ref{lemma:optimization-weighted-nash}, we can bound the competitive ratio $\rho$ of the best reply algorithm by solving the following minimization problem
\begin{align*}
\rho \leq \inf_{\substack{\lambda \in \RR_{> 0} \\ \mu \in [0,\frac{1}{d+1})}} &\Biggl\{ \frac{(d+1)\lambda}{1-(d+1)\mu} : x(x+y)^{d_r} \leq \lambda x^{{d_r}+1} + \mu y^{{d_r}+1}\\
&\quad \text{ for all $x,y \in \RR_{\geq 0}$, ${d_r} \in [0,d]$}\Biggr\} \;.
\end{align*}
Dividing the inequality by $x^{{d_r}+1}$ and substituting $z = y/x$ gives the equivalent formulation
\begin{align*}
\rho &\leq \inf_{\substack{\lambda \in \RR_{>0}\\ \mu \in [0,\frac{1}{d+1})}} \Biggl\{ \frac{(d+1)\lambda}{1-(d+1)\mu} : \lambda \geq (z+1)^{d_r} - \mu z^{{d_r}+1}\  \forall z \in \RR_{\geq 0},\ {d_r} \in [0,d]\Biggr\}\\
&= \inf_{\mu \in [0,\frac{1}{d+1})} \sup_{\substack{z \in \RR_{\geq 0}\\ {d_r} \in [0,d]}} \Biggl\{ \frac{(z+1)^{d_r} - \mu z^{{d_r}+1}}{\frac{1}{d+1}-\mu}\Biggr\}\;.
\end{align*}
Aland et al.~\cite[Lemma~5.1]{AlandDGMS2011} show that $(z+1)^d - \mu z^{d+1}  \geq (z+1)^{d_r} - \mu z^{{d_r}+1}$ for all ${d_r} \in [0,d]$ for which the latter term is non-negative. They further show \cite[Lemma~5.2]{AlandDGMS2011} that the function $(z+1)^d - \mu z^{d+1}$ has a unique maximum on $\RR_{\geq 0}$. We thus obtain
\begin{align*}
\rho \leq \inf_{\mu \in [0,\frac{1}{d+1})} \max_{z \in \RR_{\geq 0}} \Biggl\{ \frac{(z+1)^d - \mu z^{d+1}}{\frac{1}{d+1}-\mu}\Biggr\}\;.
\end{align*}

Let $\Psi_d$ be the solution to the equation $(d+1)(x+1)^d = x^{d+1}$. We claim that the term above can be upper bounded by $\Psi_d^{d+1}$. To this end, consider the functions $g : \RR_{\geq 0} \times [0,\frac{1}{d+1}) \to \RR$ and $h : \RR_{\geq 0} \to \RR$ defined as
\begin{align*}
g(x,\mu) &= \frac{(x+1)^d - \mu x^{d+1}}{\frac{1}{d+1} - \mu} &
h(x) &= \frac{d(x+1)^{d-1}}{(d+1)x^d}\;.
\end{align*}
We first show that $h(\Psi_d) \in [0,\frac{1}{d+1})$. To see this, note that
\begin{align*}
h(\Psi_d) &= \frac{d(\Psi_d+1)^{d-1}}{(d+1)\Psi_d^d} = \frac{\frac{d}{d+1}(\Psi_d+1)^{d}}{\Psi_d^d(\Psi_d+1)} = \frac{\frac{d}{(d+1)^2}\Psi_d^{d+1}}{\Psi_d^d(\Psi_d+1)} = \frac{\frac{d}{(d+1)^2}\Psi_d}{\Psi_d+1}\\
 &= \frac{1}{d+1} \cdot \frac{d}{d+1} \cdot \frac{\Psi_d}{\Psi_d+1} \in \biggl(0,\frac{1}{d+1}\biggr)\;.
\end{align*}
Second, we claim that if there is a pair $(\hat{x},\hat{\mu}) \in \RR_{\geq 0} \times (0,\frac{1}{d+1})$ with $\hat{\mu}= h(\hat{x})$, then, $g(\hat{x},\hat{\mu}) = \max_{x \in \RR_{\geq 0}} g(x, \hat{\mu})$. To see this claim, note that by construction of $h$, $\hat{x}$ satisfies the first order maximality conditions of $g(\cdot, \hat{\mu})$. As shown by Aland et al.~\cite[Lemma 5.2]{AlandDGMS2011}, the function $(x+1)^d - \mu x^{d+1}$ has a unique maximum and is increasing for values smaller than the maximum and decreasing for all values larger than the maximum. This implies the claim.

For $\hat{\mu} = h(\Psi_d)$, we obtain
\begin{align*}
\rho \leq \max_{x \in \RR_{\geq 0}} g(x, \hat{\mu}) = g(\Psi_d, \hat{\mu}) = \frac{(\Psi_d+1)^d - \hat{\mu} \Psi_d^{d+1}}{\frac{1}{d+1} - \hat{\mu}} = \Psi_d^{d+1}\;,
\end{align*}
which completes the proof.
Furthermore, we have for all $\mu \in (0,\frac{1}{d+1})$ that $\rho \geq \Psi_d^{d+1}$.
\end{proof}

\end{document}